\documentclass[reprint,amsmath,amssymb,aps,]{revtex4-2}
\usepackage{graphicx}
\usepackage{subfigure}
\usepackage{dcolumn}
\usepackage{bm}
\usepackage{amsmath}
\usepackage{amsfonts}
\usepackage{color}
\usepackage{multirow}
\usepackage{amsthm}
\usepackage{hyperref}
\newtheorem{theorem}{Theorem}[section]
\newtheorem{proposition}{Proposition}[theorem]
\newtheorem{corollary}{Corollary}[theorem]
\newtheorem{lemma}[theorem]{Lemma}

\begin{document}
\preprint{APS/123-QED}
\title{A rigorous model reduction for the anisotropic-scattering transport process}
\author{Yuan Hu}
\email{yuanhu@pku.org.cn}
\author{Chang Liu}
\email{liuchang@iapcm.ac.cn}
\author{Huayun Shen}
\email{shen\_huayun@iapcm.ac.cn}
\affiliation{Institute of Applied Physics and Computational Mathematics, Beijing, China}
\date{\today}
\begin{abstract}
In this letter, we propose a reduced-order model to bridge the particle transport mechanics and 
the macroscopic fluid dynamics in the highly scattered regime.
A rigorous mathematical derivation and a concise physical interpretation 
are presented for an anisotropic-scattering transport process 
with arbitrary order of scattering kernel.
The prediction of the theoretical model perfectly agrees with the numerical experiments.
A clear picture of the diffusion physics is revealed for the neutral particle transport in the asymptotic optically thick regime.
\end{abstract}
\maketitle

\section{Introduction}\label{Introduction}
In the statistic physics of particle transport, one objective is to construct a solid mathematical basis and 
axiomatize the multiscale hierarchy modeling of the physics of particle transport \cite{boyle2023boltzmann}.
In 1872, Boltzmann devised the Boltzmann transport equation, 
in which the positions and momenta of particles are characterized by the probability distribution function.
The Boltzmann equation bridges the microscopic Lagrangian mechanics to the mesoscopic kinetic mechanics
and reduces the degree of freedom of the system to seven dimensions, i.e., 
the physical space, the velocity space, and time.
The asymptotic theories developed by Hilbert, Chapman, and Cowling in the 1910s
provide a framework for analyzing the structure of the probability distribution function, 
based on which a low-dimensional model of the macroscopic hydrodynamic system is derived 
in the asymptotic limit of the continuum regime \cite{chapman1990mathematical},
and conscious progress has been made during the past decades \cite{hu2022relativistic,wang2020continuum}.
The asymptotic analysis reveals that the physics of momentum and energy diffusion is 
the particle momentum and energy exchange, 
and the closure modeling of the distribution function is based on 
a balance of the particle streaming and the particle collision processes. 
The macroscopic viscous and diffusion coefficients are precisely related to 
the microscopic particle interaction, including the potential and the cross-section. 
In the applications of aerodynamics and high-energy-density engineering, 
The asymptotic low-order models contribute to the development of efficient numerical methods 
for the simulation of the rarefied gas dynamics and the transport of plasma and neutral particles.
For the transport process of neutral particles, 
Larsen et al. developed the asymptotic low-order diffusion models
in the optically thick regime with an isotropic scattering
\cite{larsen1974asymptotic,larsen1987asymptotic}.
The rigorous mathematical analysis for the anisotropic-scattering transport process is still open.
In this letter, we give a rigorous mathematical derivation of the asymptotic model 
for the general anisotropic-scattering transport process. 
The mathematical derivation, physical interpretation, and numerical experiments are presented.

\section{Mathematical derivation}\label{theory}
We analyze the radiation transport equation for the anisotropic-scattering photon transport grey model here, 
while the derivation and conclusions are also appropriate for other neutral particles, such as neutrons. 
Considering a pure scattering medium and leaving out the other source terms, 
the grey thermal radiation transport equation can be written in the following scaled form,
\begin{equation} \label{Equation_2_1_}
\left\{
\begin{aligned}
&\frac{\partial I}{\partial t} + \frac{c}{\varepsilon } \vec{\Omega }\cdot \nabla I = \frac{c \sigma _{s}}{\varepsilon^{2}}  \left(S-I \right),\\
&S = \int p\left(\vec{\Omega }^{\prime} \rightarrow \vec{\Omega } \right)I \left( \vec{\Omega }^{\prime} \right)d\vec{\Omega }^{\prime}.
\end{aligned}
\right.
\end{equation}
Here $I$ is the radiation intensity, $c$ is the speed of light, 
$\sigma _{s} $ is the scattering coefficient, $S$ is the scattering source,
$\vec{\Omega }$ and $\vec{\Omega }^{\prime}$ are the outgoing and incoming direction angular variable, 
$p\left(\vec{\Omega }^{\prime} \rightarrow \vec{\Omega } \right)$ is the scattering phase function,
$\varepsilon$ is the Knudsen number, i.e., the ratio between the local mean free path and the characteristic length. 
The scattering kernel $p\left(\vec{\Omega }^{\prime} \rightarrow \vec{\Omega } \right)$ can be expanded in the Legendre polynomial space as
\begin{equation}\label{eq_legendre}
p\left(\vec{\Omega }^{\prime} \rightarrow \vec{\Omega } \right)=
\sum _{l=0}^{L} C_l P_l\left(\vec{\Omega }^{\prime}\cdot\vec{\Omega }\right),
\end{equation}
where $C_l$ is the coefficient of the $l$-order Legendre basis.

\begin{lemma} \label{source}
The scattering source $S$ is related to the moments of $I$ by
\begin{equation} \label{Equation_2_2_}
\begin{aligned}
& S = \sum_{l=0}^{L} C_{l}^{\ast} \sum_{i_1 i_2 \cdots i_l} 
\left( \Omega_{i_1 i_2 \cdots i_l}^{\left(l\right)} J_{i_1 i_2 \cdots i_l}^{\left(l\right)} \right) \\  
\end{aligned}
\end{equation}
where $\Omega_{i_1 i_2 \cdots i_l}^{\left(l\right)} = \Omega_{i_1} \Omega_{i_2}\cdots \Omega_{i_l}$,
the index $i_k\in\{x,y,z\}$ 
with the corresponding direction cosine $\Omega_{i_k} \in \{\mu,\xi,\zeta \}$ 
and $k \in \left[0,l \right]$.
$J^{\left(l \right)}$ is the $l$th-order moments of $I$, 
and $J_{i_1 i_2 \cdots i_l}^{\left(l\right)}$ is the $\{i_1 i_2 \cdots i_l\}$ component of $J^{(l)}$,
$\sum_{i_1 i_2 \cdots i_l}$ means the summation for all possible combinations of index $\{i_1 i_2 \cdots i_l\}$.
\end{lemma}

\begin{proof}
The scattering phase function \eqref{eq_legendre} can be further written as the sum of a finite series of polynomials:
\begin{equation} \label{Equation_2_3_}
\begin{aligned}
p\left(\vec{\Omega }^{\prime} \rightarrow \vec{\Omega } \right) 
= \sum _{l=0}^{L} C_{l}^{\ast} \left(\vec{\Omega }^{\prime}\cdot\vec{\Omega }\right) ^{l}
= \sum _{l=0}^{L} C_{l}^{\ast} \left(\cos\theta \right) ^{l}.
\end{aligned}
\end{equation}
Here $ C_{l}^{\ast} $ are the corresponding coefficients of $l$th polynomial $\left(\vec{\Omega }^{\prime}\cdot\vec{\Omega }\right) ^{l}$, 
$\theta$ is the scattering angle between $\vec{\Omega }^{\prime}$ and $\vec{\Omega }$. 
Thus, the scattering source $S$ can be written as:
\begin{equation} \label{Equation_2_4_}
S = \sum _{l=0}^{L} C_{l}^{\ast} \left( \int \left(\vec{\Omega }^{\prime}\cdot\vec{\Omega }\right) ^{l} I \left(\vec{\Omega }^{\prime} \right) d\vec{\Omega }^{\prime}  \right).
\end{equation}
In three-dimensional Cartesian velocity space, we have 
\begin{equation}
\begin{aligned}
\left(\vec{\Omega }^{\prime}\cdot\vec{\Omega }\right) ^{l}= & \left(\mu \mu^{\prime} + \xi \xi^{\prime} + \zeta \zeta^{\prime} \right) ^{l} 
= \sum_{i_1 i_2 \cdots i_l} \Omega_{i_1 i_2 \cdots i_l}^{\left(l\right)} 
\Omega_{i_1 i_2 \cdots i_l}^{\prime \left(l\right)}, 
\end{aligned}
\end{equation}
with 
$\Omega_{i_1 i_2 \cdots i_l}^{\prime \left(l\right)} = 
\Omega_{i_1}^{\prime} \Omega_{i_2}^{\prime}\cdots \Omega_{i_l}^{\prime}$ and
$\Omega_{i_k}^{\prime} \in \{\mu^{\prime},\xi^{\prime},\zeta^{\prime} \}$.
Therefore, the angular integration term in \eqref{Equation_2_4_} for arbitrary $l$ is:
\begin{equation} \label{Equation_2_5_}
\begin{aligned}
&\int \left(\vec{\Omega }^{\prime}\cdot\vec{\Omega }\right) ^{l} I \left(\vec{\Omega }^{\prime} \right) d\vec{\Omega }^{\prime}
= \sum_{i_1 i_2 \cdots i_l} \left( \Omega_{i_1 i_2 \cdots i_l}^{\left(l\right)} J_{i_1 i_2 \cdots i_l}^{\left(l\right)} \right), \\
\end{aligned}
\end{equation}
with 
$J_{i_1 i_2 \cdots i_l}^{\left(l\right)} = 
\int \Omega_{i_1 i_2 \cdots i_l}^{\prime \left(l\right)} 
I \left(\vec{\Omega }^{\prime} \right) d\vec{\Omega }^{\prime} $.
Substituting \eqref{Equation_2_5_} into equation \eqref{Equation_2_4_} gives equation \eqref{Equation_2_2_}.
\end{proof}

\begin{proposition} \label{S_moments}
Define $S^{\left(q \right)}$ as the $q$th-order moments of $S$.
The odd-order moments $S^{(n=odd)}$ is only related to $J^{(l=odd)}$,
and the even-order moments  $S^{(m=even)}$ is only related to $J^{(l=even)}$.
\end{proposition}

\begin{proof}
Based on Lemma \ref{source}, the components of $S^{(q)}$ is given by:
\begin{equation} \label{Equation_2_6_}
S_{i_1 i_2 \cdots i_q}^{\left(q \right)} = 
\sum_{l=0}^{L} C_{l}^{\ast} \sum_{i_1 i_2 \cdots i_l} 
\left(\left<\Omega_{i_1 i_2 \cdots i_q}^{\left(q\right)} 
\Omega_{i_1 i_2 \cdots i_l}^{\left(l\right)} \right> J_{i_1 i_2 \cdots i_l}^{\left(l\right)} \right), 
\end{equation}
where $\left< \cdot \right>$ means angular integration.
The angular integration term has the property
\begin{equation}
    \left<\mu^{r} \xi^{s} \zeta^{t} \right>
    \left\{
    \begin{aligned}
        &\neq 0, \quad r,s,t \text{ are even numbers},\\
        &=0, \quad r,s,t \text{ has odd number}.
    \end{aligned}
    \right.
\end{equation}
Therefore, we have
\begin{equation} \label{Equation_2_7_}
\left<\Omega_{i_1 i_2 \cdots i_q}^{\left(q\right)} \Omega_{i_1 i_2 \cdots i_l}^{\left(l\right)} \right> =0, \text{ for } q+l=odd.
\end{equation}
Hence, $S^{\left(n=odd \right)}$ is only related to $J^{\left(l=odd \right)}$, 
and $S^{\left(m=even \right)}$ is only related to $J^{\left(l=even \right)}$,
\begin{equation} \label{Equation_2_8_}\footnotesize
\left\{
\begin{aligned}
& S_{i_1 i_2 \cdots i_n}^{\left(n=odd \right)} =  \sum_{l=odd}^{L} C_{l}^{\ast} \sum_{i_1 i_2 \cdots i_l}^{l=odd} \left(  \left<\Omega_{i_1 i_2 \cdots i_n}^{\left(n\right)} \Omega_{i_1 i_2 \cdots i_l}^{\left(l\right)} \right> J_{i_1 i_2 \cdots i_l}^{\left(l\right)} \right) , \\
& S_{i_1 i_2 \cdots i_m}^{\left(m=even \right)} =  \sum_{l=even}^{L} C_{l}^{\ast} \sum_{i_1 i_2 \cdots i_l}^{l=even} \left(  \left<\Omega_{i_1 i_2 \cdots i_m}^{\left(m\right)} \Omega_{i_1 i_2 \cdots i_l}^{\left(l\right)} \right> J_{i_1 i_2 \cdots i_l}^{\left(l\right)} \right)   .
\end{aligned}
\right.
\end{equation}
\end{proof}

\begin{lemma} \label{even_to_rho}
The arbitrary even-order angular moments satisfy 
$J_{i_1 i_2 \cdots i_l}^{\left(l=even\right)} = \frac{\left<\Omega_{i_1 i_2 \cdots i_l}^{\left(l\right)} \right>}{\left<\Omega^{\left(0\right)} \right>} \rho$.
\end{lemma}

\begin{proof}
The zeroth moments of scattering source $S$ is
\begin{equation} \label{Equation_2_9_}
\begin{aligned}
 S^{(0)} = &\int _{\vec{\Omega }^{\prime}} \left( \int _{\vec{\Omega }} p\left(\vec{\Omega }^{\prime} \rightarrow \vec{\Omega } \right) d\vec{\Omega } \right) I \left(\vec{\Omega }^{\prime} \right) d\vec{\Omega }^{\prime}  = \rho,
\end{aligned}
\end{equation}
namely, the anisotropic scattering conserves the radiation energy.
Based on the normalization condition of the phase function, we have
\begin{equation}\label{eq_normal}
\begin{aligned}
&\int _{\vec{\Omega }^{\prime}} p\left(\vec{\Omega }^{\prime} \rightarrow \vec{\Omega } \right)d\vec{\Omega }^{\prime}
= \sum_{l}^{L} C_{l}^{\ast} \int _{\vec{\Omega }^{\prime}} \left(\vec{\Omega }^{\prime} \cdot \vec{\Omega } \right)^{l} d\vec{\Omega }^{\prime} \\
&= \sum_{l=even}^{L} C_{l}^{\ast} \sum_{i_1 i_2 \cdots i_{l}}^{l=even} \left( \left<\Omega_{i_1 i_2 \cdots i_{l}}^{\prime \left(l\right)} \right> \Omega_{i_1 i_2 \cdots i_{l}}^{ \left(l\right)} \right)
=1.
\end{aligned}
\end{equation}
Taking angular integration to equation \eqref{eq_normal}, we have 
\begin{equation} \label{Equation_2_9.1_}
\begin{aligned}
& \sum_{l=even}^{L} C_{l}^{\ast} \sum_{i_1 i_2 \cdots i_{l}}^{l=even} \left( \left<\Omega_{i_1 i_2 \cdots i_{l}}^{\prime \left(l\right)} \right> \left< \Omega_{i_1 i_2 \cdots i_{l}}^{ \left(l\right)} \right> \right) 
= \left<\Omega^{\left(0 \right)} \right>,
\end{aligned}
\end{equation}
which gives
\begin{equation}\label{eq_normal1}
    \sum_{l=even}^{L} C_{l}^{\ast} \sum_{i_1 \cdots i_{l}}^{l=even} \left( \left<\Omega_{i_1 \cdots i_{l}}^{ \left(l\right)} \right> \frac{\left< \Omega_{i_1 \cdots i_{l}}^{ \left(l\right)} \right>}{\left<\Omega^{\left(0 \right)} \right>}  \right) =1.
\end{equation}
Base on Proposition \ref{S_moments} and equation \eqref{Equation_2_9_}, we have
\begin{equation}\label{eq_s0}
    \rho=S^{(0)} = 
    \sum_{l=even}^{L} C_{l}^{\ast} \sum_{i_1 i_2 \cdots i_l}^{l=even} 
    \left(   \left<\Omega_{i_1 i_2 \cdots i_l}^{\left(l\right)} \right> 
    J_{i_1 i_2 \cdots i_l}^{\left(l\right)} \right),
\end{equation}
which leads to
\begin{equation}\label{eq_s1}
    \sum_{l=even}^{L} C_{l}^{\ast} \sum_{i_1 i_2 \cdots i_l}^{l=even} \left(   \left<\Omega_{i_1 i_2 \cdots i_l}^{\left(l\right)} \right> \frac{J_{i_1 i_2 \cdots i_l}^{\left(l\right)}}{\rho} \right)=1.
\end{equation}
Thus, equations \eqref{eq_normal1} and \eqref{eq_s1} gives 
\begin{equation}
    J_{i_1 i_2 \cdots i_l}^{\left(l=even\right)} = \frac{\left< \Omega_{i_1 i_2 \cdots i_{l}}^{ \left(l\right)} \right>}{\left<\Omega^{\left(0 \right)} \right>} \rho.
\end{equation}
\end{proof}

\begin{corollary}\label{isotropic}
    If the scattering kernel $p\left(\vec{\Omega }^{\prime} \rightarrow \vec{\Omega } \right)$ contains only even-order terms, the scattering process preserves the isotropic distribution of photons.
\end{corollary}
\begin{proof}
If the scattering kernel \eqref{eq_legendre} contains only even-order terms, the scattering source $S$ satisfies
\begin{equation}\footnotesize
\begin{aligned}
 S =& \sum_{l=even}^{L} C_{l}^{\ast} \sum_{i_1 i_2 \cdots i_l} \left( \Omega_{i_1 i_2 \cdots i_l}^{\left(l\right)} J_{i_1 i_2 \cdots i_l}^{\left(l\right)} \right) \\
= &\frac{\rho}{\left<\Omega^{\left(0\right)} \right>} 
\sum_{l=even}^{L} C_{l}^{\ast} \sum_{i_1 i_2 \cdots i_l} \left( \Omega_{i_1 i_2 \cdots i_l}^{\left(l\right)} 
\left<\Omega_{i_1 i_2 \cdots i_l}^{\left(l\right)} \right> \right).
\end{aligned}
\end{equation}
The normalization property \eqref{eq_normal} states that $S =  {\rho}/{\left<\Omega^{\left(0\right)} \right>}$,
which shows that the scattering process preserves the isotropic distribution of photons.
\end{proof}

Lemma \ref{source}, Proposition \ref{S_moments}, Lemma \ref{even_to_rho}, and Corollary \ref{isotropic} 
hold for arbitrary order $L$ of collision kernel \eqref{eq_legendre} and 
all flow regimes.
In the following, we close the energy flux $J^{\left(1 \right)}$ in the optically thick regime.
According to the asymptotic theory \cite{chapman1990mathematical},
the radiant intensity can be formally expanded as
\begin{equation} \label{Equation_2_12_}
I = S - \varepsilon \frac{1}{\sigma _{s}} \vec{\Omega }\cdot \nabla S+O(\varepsilon^2).
\end{equation}
According to Proposition \ref{S_moments}, the radiant energy flux $J^{\left(1 \right)}$ has the form
\begin{equation} \label{Equation_2_14_}\footnotesize
\begin{aligned}
 J^{\left(1 \right)} =&
S^{\left(1 \right)} - \varepsilon \frac{1}{\sigma _{s}} \nabla\cdot  S^{\left(2 \right)}+O(\varepsilon^2)\\
=& \sum_{l=odd}^{L} C_{l}^{\ast} \sum_{i_1 i_2 \cdots i_l}^{l=odd} \left(  \left<\Omega_{}^{\left(1\right)} \Omega_{i_1 i_2 \cdots i_l}^{\left(l\right)} \right> J_{i_1 i_2 \cdots i_l}^{\left(l\right)} \right) \\
&- \varepsilon \frac{1}{\sigma _{s}} \nabla\cdot \sum_{l=even}^{L} C_{l}^{\ast} \sum_{i_1 i_2 \cdots i_l}^{l=even} \left( \left<\Omega_{}^{\left(2\right)} \Omega_{i_1 i_2 \cdots i_l}^{\left(l\right)} \right> J_{i_1 i_2 \cdots i_l}^{\left(l\right)} \right) 
\end{aligned}
\end{equation}
To close the radiant energy flux $J^{\left(1 \right)}$ up to order $O(\varepsilon)$, 
we need to calculate the $O(1)$ and $O(\varepsilon)$ contributions of $J^{\left(n=odd \right)}$ and
the $O(\varepsilon)$ contribution of $J^{\left(m=even \right)}$ to $J^{\left(1 \right)}$, i.e.,
\begin{equation} \label{Equation_2_11_}\footnotesize
\left\{
\begin{aligned}
&J^{\left(m=even \right)} = S^{\left(m \right)} - \varepsilon \frac{1}{\sigma _{s}} \nabla\cdot  S^{\left(m+1 \right)} \approx S^{\left(m \right)} , \\
& J^{\left(n=odd \right)} = S^{\left(n \right)} - \varepsilon \frac{1}{\sigma _{s}} \nabla\cdot  S^{\left(n+1 \right)} = S^{\left(n \right)} - \varepsilon \frac{1}{\sigma _{s}} \nabla\cdot  J^{\left(n+1 \right)} .
\end{aligned}
\right.
\end{equation}

\begin{lemma} \label{odd_to_J1}
The following recurrence relation holds for arbitrary odd-order angular moments,
\begin{equation} \label{Equation_2_15}\footnotesize
\left\{
\begin{aligned}
& \left.
J_{i_1 i_2 \cdots i_l}^{\left(l=odd\right)} = \frac{ \left<\mu \Omega_{i_1 i_2 \cdots i_l}^{\left(l\right)} \right>}{\left<\mu^{2} \right>} J_{x}^{(1)}
\right|_{i_1 i_2 \cdots i_l = x^{odd}y^{even}z^{even}}, \\
& \left.
J_{i_1 i_2 \cdots i_l}^{\left(l=odd\right)} = \frac{ \left<\xi \Omega_{i_1 i_2 \cdots i_l}^{\left(l\right)} \right>}{\left<\mu^{2} \right>} J_{y}^{(1)}
\right|_{i_1 i_2 \cdots i_l = x^{even}y^{odd}z^{even}}, \\
& \left.
J_{i_1 i_2 \cdots i_l}^{\left(l=odd\right)} = \frac{ \left<\zeta \Omega_{i_1 i_2 \cdots i_l}^{\left(l\right)} \right>}{\left<\mu^{2} \right>} J_{z}^{(1)}
\right|_{i_1 i_2 \cdots i_l = x^{even}y^{even}z^{odd}}, \\
& \left.
J_{i_1 i_2 \cdots i_l}^{\left(l=odd\right)} = 0
\right|_{i_1 i_2 \cdots i_l = x^{odd}y^{odd}z^{odd}}\\
\end{aligned}
\right.
\end{equation}
\end{lemma}

\begin{proof}
According to equation \eqref{Equation_2_11_}, Lemma \ref{S_moments}, Proposition \ref{S_moments} and Lemma \ref{even_to_rho}, 
the components of odd-order moments $J^{(l=odd)}$ follow the equation system \eqref{Equation_2_17_},
\begin{equation}\label{Equation_2_17_}\footnotesize
\begin{aligned}
&\left[\delta_{ql} -  C_{l}^{\ast} \left<\Omega_{i_1 i_2 \cdots i_q}^{\left(q\right)} \Omega_{i_1 i_2 \cdots i_l}^{\left(l\right)} \right> \right] \left(J_{i_1 i_2 \cdots i_l}^{\left(l\right)} \right) = \\
&\left\{
\begin{aligned}
&\left. 
-\varepsilon \frac{1}{\sigma _{s}} \frac{ \left<\mu \Omega_{i_1 i_2 \cdots i_q}^{\left(q \right)} \right> }{\left<\Omega^{(0)} \right>} \frac{\partial}{\partial x} \rho
\right|_{ \begin{subarray}{l}
i_1 i_2 \cdots i_q = x^{odd}y^{even}z^{even} \\
i_1 i_2 \cdots i_l = x^{odd}y^{even}z^{even}
\end{subarray} },  \\
&
\left. 
-\varepsilon \frac{1}{\sigma _{s}} \frac{ \left<\xi \Omega_{i_1 i_2 \cdots i_q}^{\left(q \right)} \right> }{\left<\Omega^{(0)} \right>} \frac{\partial}{\partial y} \rho
\right|_{ \begin{subarray}{l}
i_1 i_2 \cdots i_q = x^{even}y^{odd}z^{even} \\
i_1 i_2 \cdots i_l = x^{even}y^{odd}z^{even}
\end{subarray} },  \\ 
&
\left. 
-\varepsilon \frac{1}{\sigma _{s}} \frac{ \left<\zeta \Omega_{i_1 i_2 \cdots i_q}^{\left(q \right)} \right> }{\left<\Omega^{(0)} \right>} \frac{\partial}{\partial z} \rho
\right|_{ \begin{subarray}{l}
i_1 i_2 \cdots i_q = x^{even}y^{even}z^{odd} \\
i_1 i_2 \cdots i_l = x^{even}y^{even}z^{odd}
\end{subarray} },  \\
& 
\left. 0 \right|_{\begin{subarray}{l}
i_1 i_2 \cdots i_q = x^{odd}y^{odd}z^{odd} \\
i_1 i_2 \cdots i_l = x^{odd}y^{odd}z^{odd}
\end{subarray}},  \\
\end{aligned}
\right.    
\end{aligned}
\end{equation}
where $\delta_{ql} = 1$ 
if $\Omega_{i_1 i_2 \cdots i_q}^{\left(q\right)} = \Omega_{i_1 i_2 \cdots i_l}^{\left(l\right)}$
and $\delta_{ql} = 0$ 
if $\Omega_{i_1 i_2 \cdots i_q}^{\left(q\right)} \neq \Omega_{i_1 i_2 \cdots i_l}^{\left(l\right)}$.
Solving equation system \eqref{Equation_2_17_} implies that
\begin{equation}
J_{i_1 i_2 \cdots i_l}^{\left(l\right)} \propto
\left\{
\begin{aligned}
& \left. J_{x}^{(1)} \right|_{i_1 i_2 \cdots i_l = x^{odd}y^{even}z^{even}} \\
& \left. J_{y}^{(1)} \right|_{i_1 i_2 \cdots i_l = x^{even}y^{odd}z^{even}}\\
& \left. J_{z}^{(1)} \right|_{i_1 i_2 \cdots i_l = x^{even}y^{even}z^{odd}}\\
& \left. 0 \right|_{i_1 i_2 \cdots i_l = x^{odd}y^{odd}z^{odd}} \\
\end{aligned}
\right..
\end{equation}
According to Lemma \ref{even_to_rho} and equation \eqref{Equation_2_11_}, 
the relationship between $J_{i_1 i_2 \cdots i_l}^{\left(l=odd\right)} $ and $J_{x}^{(1)}$ could be obtained
\begin{equation}\label{Equation_2_19_}
\begin{aligned}
&\frac{J_{i_1 i_2 \cdots i_l}^{\left(l=odd\right)}}{J_{x}^{(1)}} =
\frac{S_{i_1 i_2 \cdots i_l}^{\left(l=odd\right)}}{S_{x}^{\left(1\right)}} =
\frac{\frac{\partial}{\partial x} J_{xi_1 i_2 \cdots i_l}^{\left(l+1\right)}}{\frac{\partial}{\partial x} J_{xx}^{\left(2\right)}} 
= \frac{ \left<\mu \Omega_{i_1 i_2 \cdots i_l}^{\left(l \right)} \right> }{\left<\mu^{2} \right>},
\end{aligned}
\end{equation}
for $i_1 i_2 \cdots i_l = x^{odd}y^{even}z^{even}$.
Similar relation can be derived for $i_1 i_2 \cdots i_l = x^{even}y^{odd}z^{even}$ and $i_1 i_2 \cdots i_l = x^{even}y^{even}z^{odd}$.
\end{proof}

\begin{theorem} \label{diffusion}
The asymptotic equation for the radiative transfer equation \eqref{Equation_2_1_} 
in the optically thick limiting regime is
\begin{equation} \label{Equation_2_20_}
\frac{\partial \rho}{\partial t} - \nabla\cdot D \nabla \rho = 0, 
\end{equation}
where the diffusion coefficient follows
\begin{equation}\label{eq_diffcoef}
    D=\frac{c}{3\left(1 - \overline{\cos\theta} \right) \sigma_s}, 
\end{equation}
and $\overline{\cos\theta}$ is the average cosine of the scattering angle.
\end{theorem}
\begin{proof}
According to the normalization condition of phase function and \eqref{eq_normal1}, we have
\begin{equation} \label{Equation_2_21_}
\sum_{l=even}^{L} C_{l}^{\ast} \left< \left( \cos\theta \right)^{l} \right> =1
= \sum_{l=even}^{L} C_{l}^{\ast} \sum_{i_1 i_2 \cdots i_l}^{l=even} \left( \frac{\left<\Omega_{i_1 i_2 \cdots i_l}^{\left(l\right)} \right>^{2}}{\left<\Omega^{\left(0\right)} \right>} \right)
\end{equation}
which implies
\begin{equation}
    \left< \left( \cos\theta \right)^{l} \right> =
 \sum_{i_1 i_2 \cdots i_l}^{l=even} \left( \frac{\left<\Omega_{i_1 i_2 \cdots i_l}^{\left(l\right)} \right>^{2}}{\left<\Omega^{\left(0\right)} \right>} \right).
\end{equation}
Furthermore,
\begin{equation} \label{Equation_2_22_}\footnotesize
\begin{aligned}
& \left< \left( \cos\theta \right)^{l} \right> = \sum_{i_1 i_2 \cdots i_l}^{l=even} \left(  \frac{\left<\Omega_{i_1 i_2 \cdots i_l}^{\left(l\right)} \right>^{2}}{\left<\Omega^{\left(0\right)} \right>} \right)  
 = \sum_{i_1 i_2 \cdots i_{l-1}}^{l=even} \left( \frac{\left<\mu \Omega_{i_1 i_2 \cdots i_{l-1}}^{\left(l-1\right)} \right>^{2}}{\left<\Omega^{\left(0\right)} \right>} \right) \\
& + \sum_{i_1 i_2 \cdots i_{l-1}}^{l=even} \left( \frac{\left<\xi \Omega_{i_1 i_2 \cdots i_{l-1}}^{\left(l-1\right)} \right>^{2}}{\left<\Omega^{\left(0\right)} \right>} \right) +
\sum_{i_1 i_2 \cdots i_{l-1}}^{l=even} \left( \frac{\left<\zeta \Omega_{i_1 i_2 \cdots i_{l-1}}^{\left(l-1\right)} \right>^{2}}{\left<\Omega^{\left(0\right)} \right>} \right) \\
& = 3 \sum_{i_1 i_2 \cdots i_{l-1}}^{l=even} \left( \frac{\left<\mu \Omega_{i_1 i_2 \cdots i_{l-1}}^{\left(l-1\right)} \right>^{2}}{\left<\Omega^{\left(0\right)} \right>} \right) =
\sum_{i_1 i_2 \cdots i_{l-1}}^{l=even} \left( \frac{\left<\mu \Omega_{i_1 i_2 \cdots i_{l-1}}^{\left(l-1\right)} \right>^{2}}{\left<\mu^{2} \right>} \right).
\end{aligned}
\end{equation}
Thus based on Proposition \ref{S_moments}, equation \eqref{Equation_2_11_}, Lemma \ref{even_to_rho}, Lemma \ref{odd_to_J1} and equation \eqref{Equation_2_22_}, the energy flux $J_{x}^{(1)}$ can be derived as
\begin{equation} \label{Equation_2_23_}\footnotesize
\begin{aligned}
& J_{x}^{(1)} = \sum_{l=odd}^{L} C_{l}^{\ast} \sum_{i_1 i_2 \cdots i_{l}}^{l=odd} \left(  \left<\mu \Omega_{ i_1 i_2 \cdots i_{l}}^{\left(l\right)} \right> J_{i_1 i_2 \cdots i_{l}}^{\left(l\right)} \right)
 - \varepsilon \frac{1}{\sigma _{s}}  \frac{\partial}{\partial x} J_{xx}^{\left(2\right)}  \\
& = J_{x}^{(1)} \sum_{l=odd}^{L} C_{l}^{\ast} \sum_{i_1 i_2 \cdots i_{l}}^{l=odd} \left( \frac{\left<\mu \Omega_{i_1 i_2 \cdots i_{l}}^{\left(l\right)} \right>^{2}}{\left<\mu^{2} \right>} \right)
 - \varepsilon \frac{1}{\sigma _{s}} \frac{\left<\mu^{2} \right>}{\left<\Omega^{(0)} \right>} \frac{\partial}{\partial x} \rho \\
& = \overline{\cos\theta} J_{x}^{(1)} - \varepsilon \frac{1}{\sigma _{s}} \frac{1}{3} \frac{\partial}{\partial x} \rho ,
\end{aligned}
\end{equation}
which implies 
\begin{equation}\label{eq_flux}
J_{x}^{(1)} = - \varepsilon \frac{1}{3\left(1-\overline{\cos\theta} \right)\sigma _{s}} 
\frac{\partial}{\partial x} \rho,
\end{equation}
where $\overline{\cos\theta}$ is the average cosine of the scattering angle:
\begin{equation} \label{Equation_2_24_}
\overline{\cos\theta} = \int _{\vec{\Omega }^{\prime}} \cos\theta \cdot p\left(\vec{\Omega }^{\prime} \rightarrow \vec{\Omega } \right)d\vec{\Omega }^{\prime} = \sum_{l=odd}^{L} C_{l}^{\ast} \left< \left( \cos\theta \right)^{l+1} \right>.
\end{equation}
A routine calculation shows that $J_{y}^{(1)}$ and $J_{z}^{(1)}$ have the same form as $J_{x}^{(1)}$ in \eqref{Equation_2_23_} and \eqref{eq_flux}, which gives
\begin{equation} \label{Equation_2_25_}
J^{(1)} = - \varepsilon \frac{1}{3\left(1-\overline{\cos\theta} \right)\sigma _{s}} \nabla \rho.  
\end{equation}
The radiative transfer equation \eqref{Equation_2_1_} is reduced to 
a diffusion equation \eqref{Equation_2_20_} with the diffusive flux \eqref{Equation_2_25_} 
and the diffusion coefficient \eqref{eq_diffcoef}.
\end{proof}

\section{Physical interpretation} \label{physics}
The rigorous mathematical derivation is consistent with the physical picture of the anisotropic scattering process \cite{lee2020nuclear}.
The average scattering angle $\phi$ can be calculated by $\phi = \arccos \left(\overline{\cos\theta}\right)$.
Therefore, the particle transport mean free path $\lambda_{tr}$ can be estimated by
\begin{equation} \label{Equation_2_26_}
\begin{aligned}
 \lambda_{tr} & = \lambda_{s} + \lambda_{s} \cos\phi + \lambda_{s} \left( \cos\phi \right) ^{2} + \cdots  = \frac{\lambda_{s}}{1 - \overline{\cos\theta}},    
\end{aligned}
\end{equation}
where $\lambda_{s} = 1/{\sigma_{s}}$ is the scattering mean free path.
Analogy to the isotropic scattering, the diffusion coefficient is proportional to transport means free path, i.e.
\begin{equation}\label{eq_phyD}
\begin{aligned}
D = \frac{\lambda_{tr} c}{3} = \frac{c}{3\left(1 - \overline{\cos\theta} \right) \sigma_{s}},    
\end{aligned}
\end{equation}
which is consistent with Theorem \ref{diffusion}.
The reduced model \eqref{Equation_2_20_} is valid for arbitrary order of phase function, 
which reveals the diffusion physics of neutral particle transport in the asymptotic optically thick regime.

\section{Numerical experiments}\label{examples}
\begin{figure}
  \centering
  \subfigure[]{\includegraphics[width=0.45\textwidth]{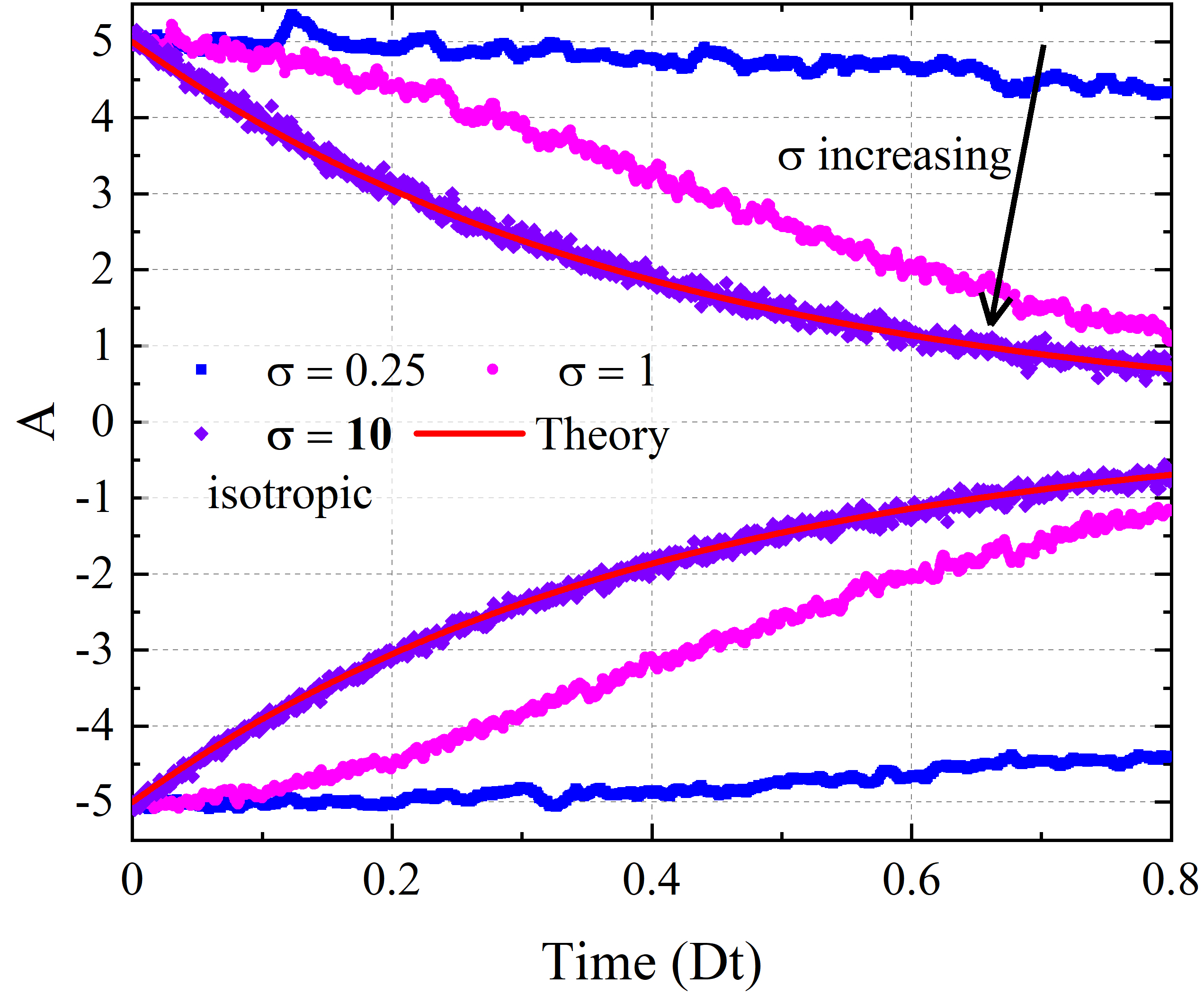}}
  \subfigure[]{\includegraphics[width=0.45\textwidth]{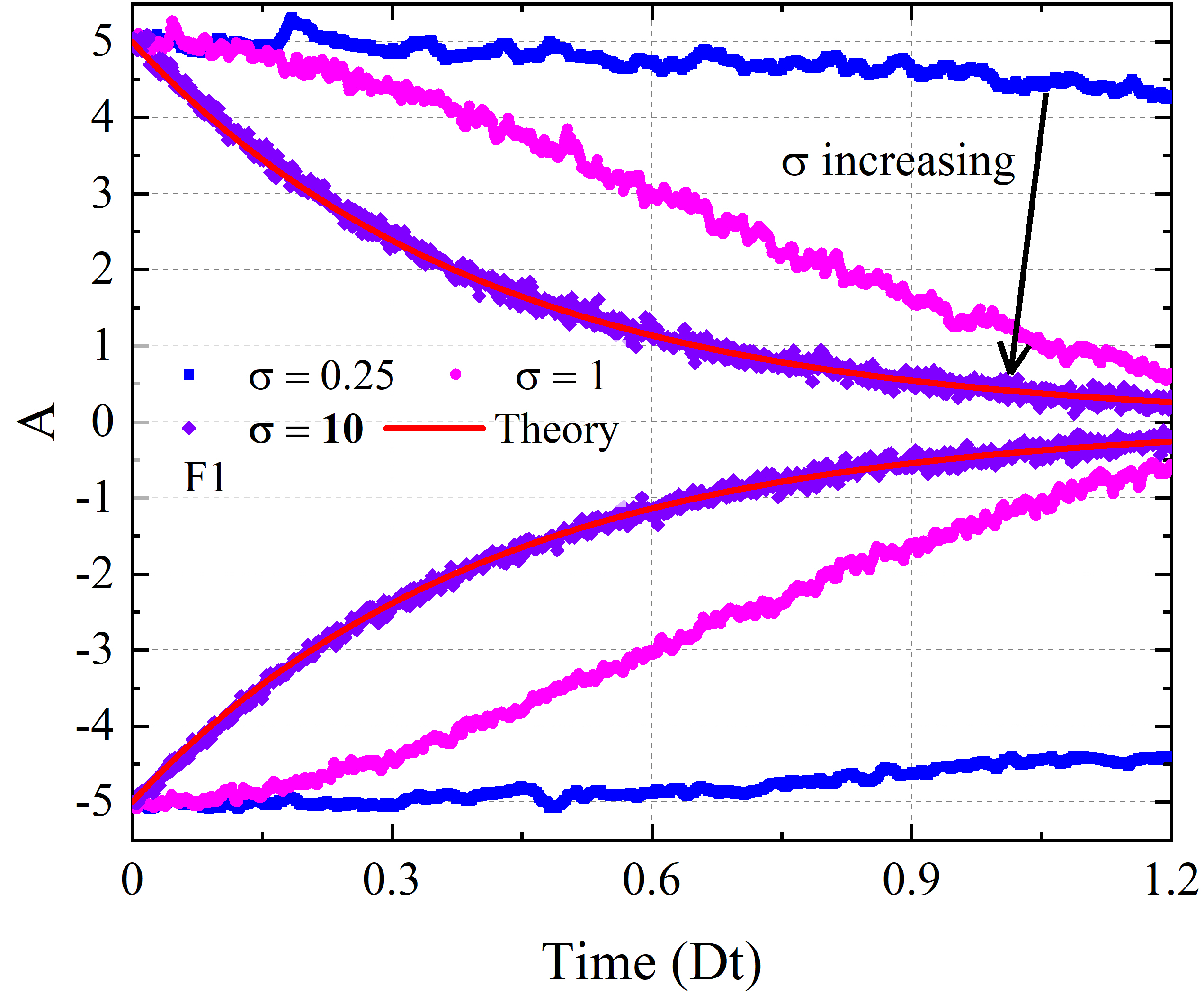}}
  \caption{The MC simulated results of the isotropic and F1 scattering cases under different $\sigma$, compared with the theory curve. (a) isotropic scattering, (b) F1 scattering.}
  \label{fig_image1}
\end{figure}

\begin{figure}
  \centering
  \includegraphics[width=0.45\textwidth]{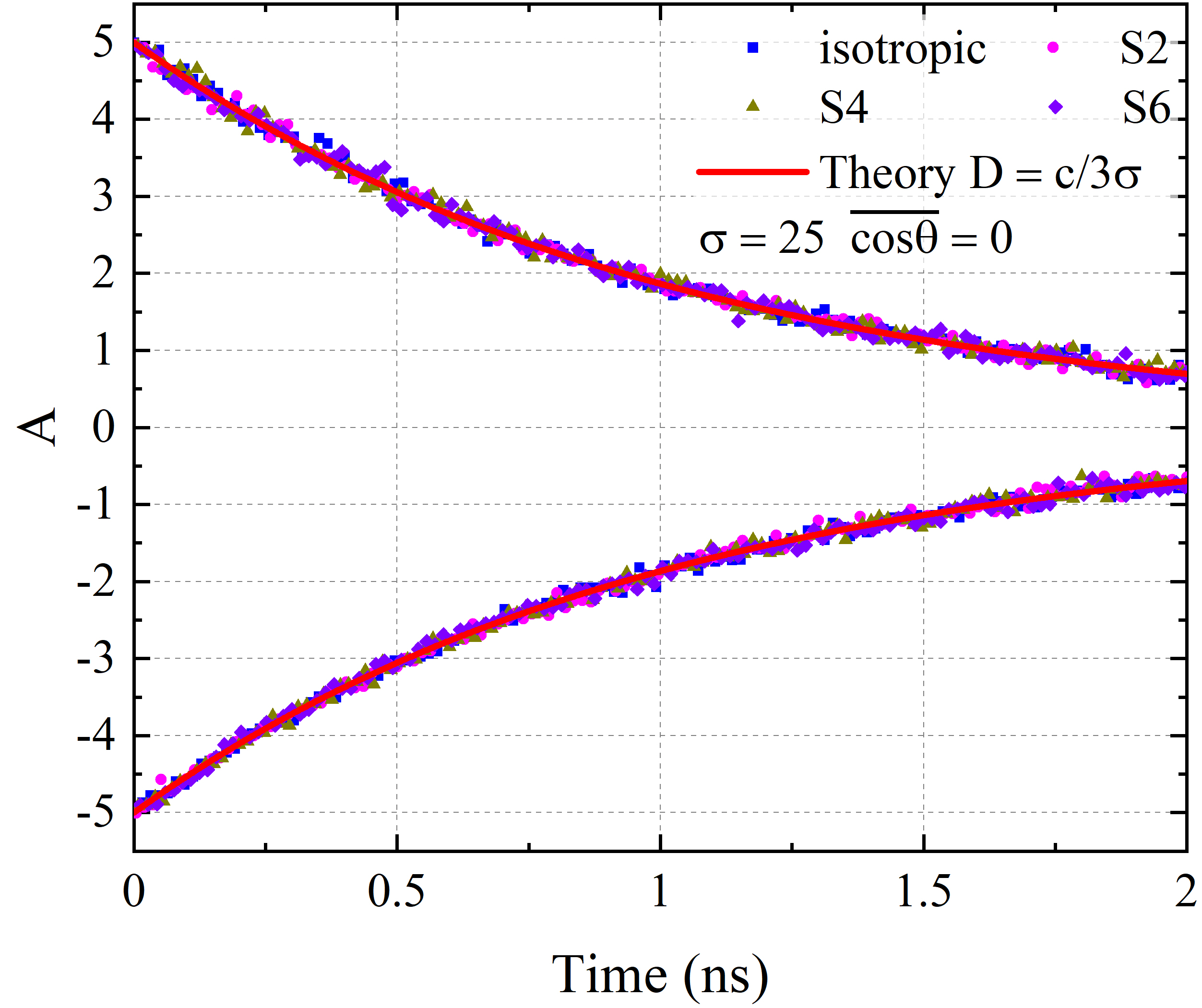}
  \caption{The comparison between the MC simulated results and the theory curve for different scattering phase functions with $\overline{\cos\theta} = 0$.}
  \label{fig_image2}
\end{figure}
\begin{figure}
  \centering
  \includegraphics[width=0.45\textwidth]{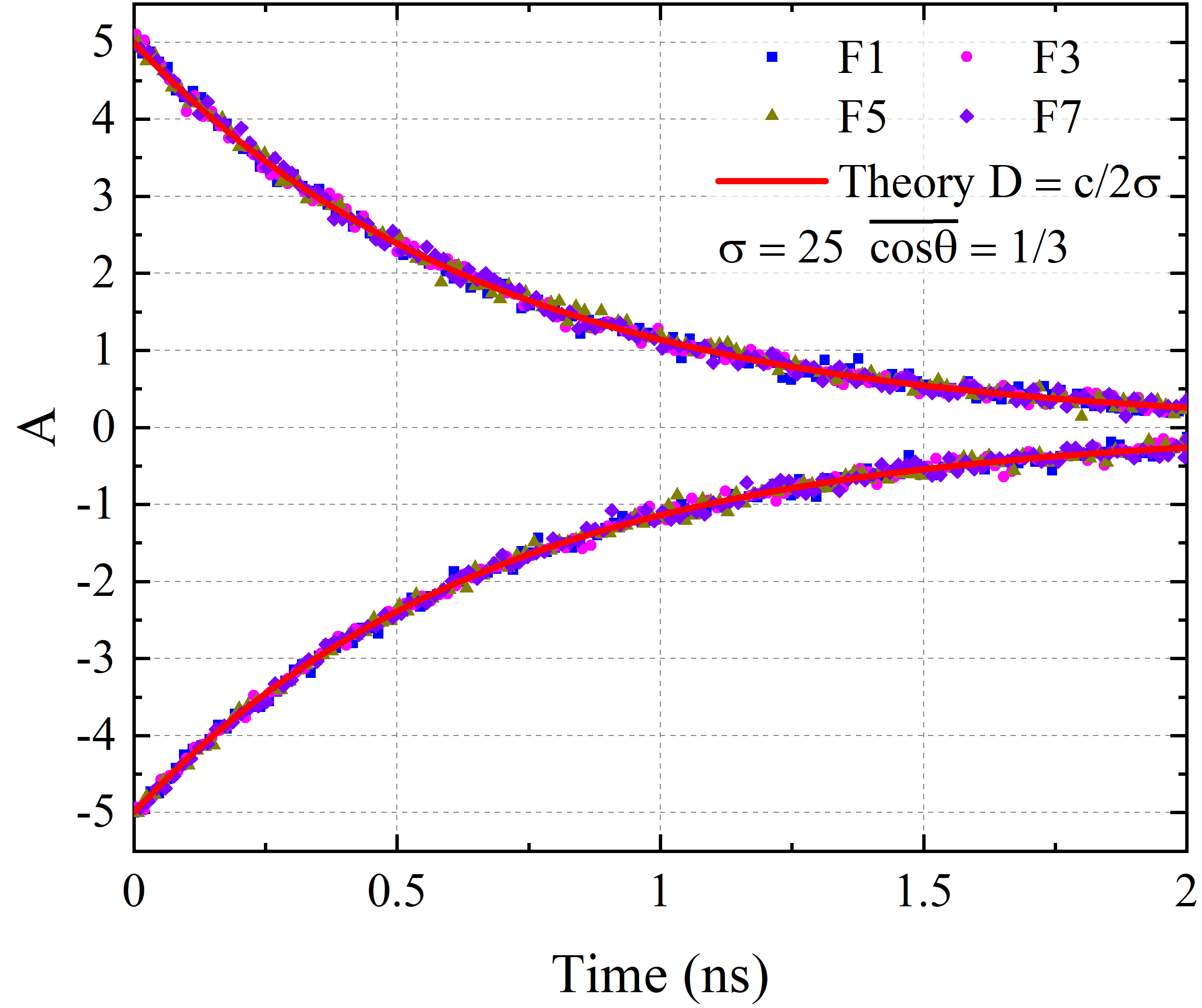}
  \caption{The comparison between the MC simulated results and the theory curve for different, forward scattering phase functions with $\overline{\cos\theta} = 1/3$.}
  \label{fig_image3}
\end{figure}
\begin{figure}
  \centering
  \includegraphics[width=0.45\textwidth]{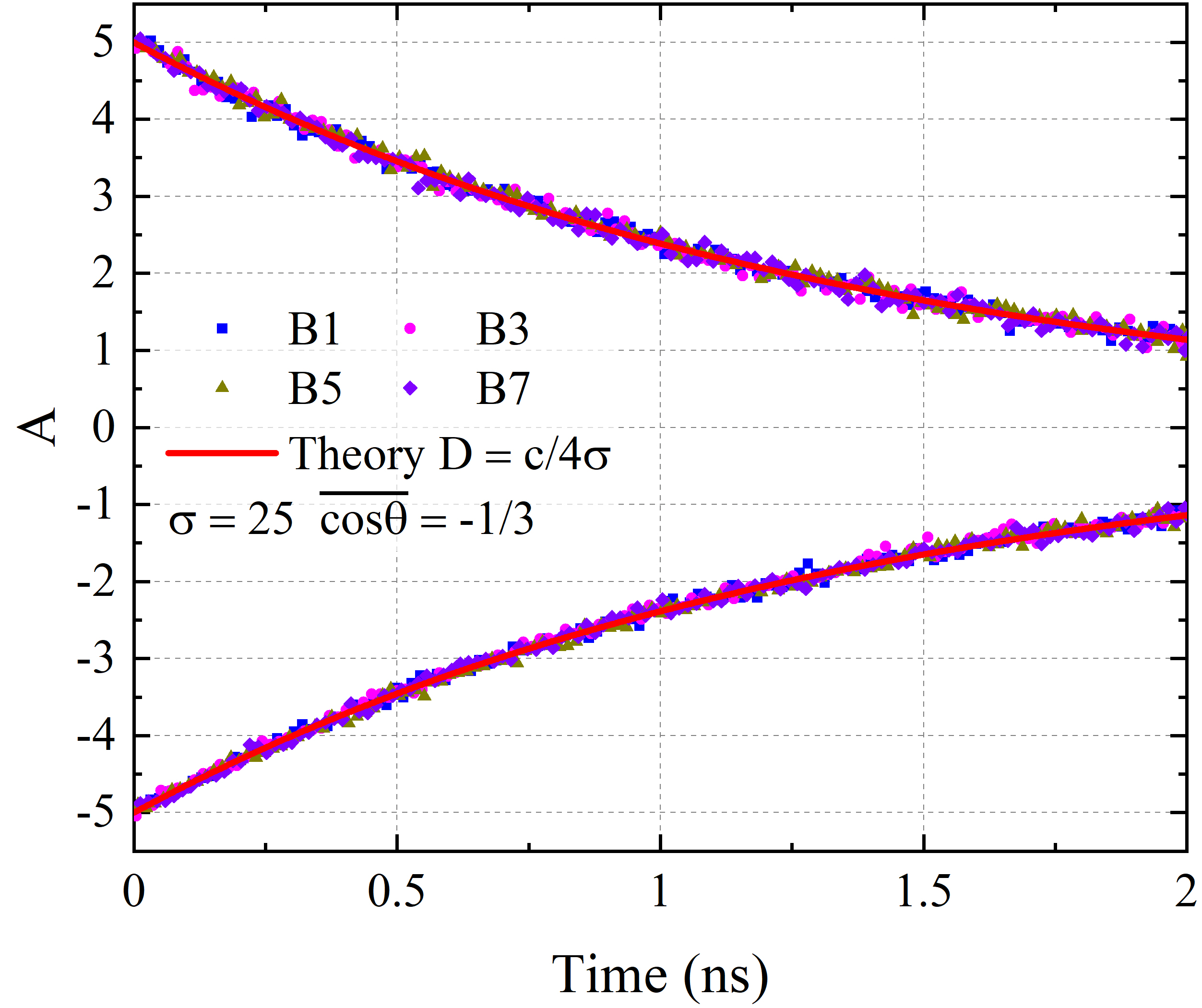}
  \caption{The comparison between the MC simulated results and the theory curve for different backward scattering phase functions with $\overline{\cos\theta} = -1/3$.}
  \label{fig_image4}
\end{figure}
Numerical experiments are presented to validate the Theorem \ref{diffusion}. 
A pure scattering process is considered in a one-dimensional geometry space with periodic boundary conditions.
The initial radiation energy distribution of $\rho$ is
$\rho = 10 + 5 \sin\left(\frac{\pi}{2} x\right), x\in \left[-2,2 \right]$. 
The geometry space has been divided into $10^3$ meshes, and the total simulated particles are $10^7$.
Based on \eqref{Equation_2_20_}, the theoretical prediction of the time evolution of the amplitude follows:
\begin{equation} \label{Equation_3_1_}
A\left(x,t \right) = \rho \left(x,t \right) - 10 = 5\sin\left(\frac{\pi}{2} x\right) e^{-D\left(\frac{\pi}{2} \right)^{2}t} .
\end{equation}
We compare the time evolution of amplitude at peak $A(-1,t)$ and valley $A(1,t)$ between the theoretical and numerical results.
The phase functions are listed in tables \ref{table1}-\ref{table3}, 
where $SL$, $FL$, and $BL$ stand for $L$-th order symmetrical, forward and backward scattering.
Figure \ref{fig_image1} compares the theoretical predictions and kinetic simulation results of the isotropic and F1 scattering with different $\sigma$, 
showing a clear convergence to the theoretical curve as $\sigma$ increases.
An excellent agreement is observed with $\sigma = 10$.
Various phase functions are tested as shown in tables (\ref{table1}-\ref{table3}).
The kinetic simulation results with scattering parameters
$\overline{\cos\theta} = 0$, 
$\overline{\cos\theta} = 1/3$ and 
$\overline{\cos\theta} = -1/3$, with 
$\sigma = 25$,
are given in Figure \ref{fig_image2}-\ref{fig_image4}.
Perfect agreement between the theoretical prediction and kinetic simulation data is observed, which validates our reduced model in the optically thick regime.

\begin{table}
    \centering
    \begin{tabular}{ c c c c }
    \hline
    \hline
     symbol & phase function & $\overline{\cos\theta}$ & D \\  \hline
     isotropic & $1/2$ & \multirow{4}{*}{$0$} & \multirow{4}{*}{$c/3\sigma$} \\  
     $S2$ & $3{\cos\theta}^2 /2$ & ~ & ~  \\  
     $S4$ & $5{\cos\theta}^4 /2$ & ~ & ~  \\  
     $S6$ & $7{\cos\theta}^6 /2$ & ~ & ~  \\  \hline \hline
    \end{tabular}
    \caption{The corresponding average cosine of the scattering angle and diffusion coefficient of symmetrical scattering phase functions.}
    \label{table1}
\end{table}

\begin{table} 
    \centering
    \begin{tabular}{ c c c c }
    \hline
    \hline
     symbol & Legendre coefficients & $\overline{\cos\theta}$ & D \\  \hline
     $F1$ & $C_{l=odd} = C_{l=even} = 1/2$ & \multirow{4}{*}{$1/3$} & \multirow{4}{*}{$c/2\sigma$} \\  
     $F3$ & $C_{l=odd} = C_{l=even} = 1/4$ & ~ & ~  \\  
     $F5$ & $C_{l=odd} = C_{l=even} = 1/6$ & ~ & ~  \\  
     $F7$ & $C_{l=odd} = C_{l=even} = 1/8$ & ~ & ~  \\  \hline \hline
    \end{tabular}
    \caption{The corresponding average cosine of the scattering angle and diffusion coefficient of forward scattering phase functions.}
    \label{table2}
\end{table}

\begin{table} 
    \centering
    \begin{tabular}{ c c c c }
    \hline
    \hline
     symbol & Legendre coefficients & $\overline{\cos\theta}$ & D \\  \hline
     $B1$ & $- C_{l=odd} = C_{l=even} = 1/2$ & \multirow{4}{*}{$-1/3$} & \multirow{4}{*}{$c/4\sigma$} \\  
     $B3$ & $- C_{l=odd} = C_{l=even} = 1/4$ & ~ & ~  \\  
     $B5$ & $- C_{l=odd} = C_{l=even} = 1/6$ & ~ & ~  \\  
     $B7$ & $- C_{l=odd} = C_{l=even} = 1/8$ & ~ & ~  \\  \hline \hline
    \end{tabular}
    \caption{The corresponding average cosine of the scattering angle and diffusion coefficient of backward scattering phase functions.}
    \label{table3}
\end{table}

\section{Conclusion}\label{conclusions}
In the transport process of neutral particle with anisotropic scattering, the even-order scattering kernel will not change the isotropy of the particle velocity distribution.
In the highly scattered regime, the high-order transport equation reduces to a low-order diffusion equation.
Based on rigorous mathematical derivations, the diffusion coefficient is derived for an arbitrary-order scattering kernel.
The numerical experiments show excellent agreement with the theoretical prediction, and a clear physical picture of diffusion is revealed.
The theoretical results in this letter contribute to the understanding and simulation of the neutral particle transport physics in the asymptotic optically thick regime.

\begin{acknowledgments}
We thank Dr. Yanli Wang from Beijing Computational Science Research Center for helpful discussions and for providing computational resources. 
Chang Liu is partially supported by the National Natural Science Foundation of China (12102061,12031001),
the National Key R\&D Program of China (2022YFA1004500),
the Presidential Foundation of the China Academy of Engineering Physics (YZJJZQ2022017),
and the National Key R\&D Program of China (2022YFA1004500).
\end{acknowledgments}


\end{document}